\documentclass[preprint,12pt]{elsarticle}
\usepackage{ifthen}
\usepackage{latexsym}
\usepackage{amsfonts}
\usepackage{amssymb}
\usepackage{amsthm}
\usepackage{amsmath}
\usepackage{xspace,comment,color}
\usepackage{graphicx}
\usepackage{subfigure}
\def\rootfig{./}

\usepackage{mydef,boldfonts} 

\begin{document}
\journal{European Journal of Mechanics B}

\title{Remarks on the stability of the Navier-Stokes equations
  supplemented with stress-free boundary conditions}
\author{J.-L. Guermond$^{\text{1,2}}$, J. L\'eorat$^{\text{3}}$, F. Luddens$^{\text{1,2}}$ and C. Nore$^{\text{1,4}}$}

\address{$^{\text{1}}$Laboratoire d'Informatique pour la
  M\'ecanique et les Sciences de l'Ing\'enieur, CNRS UPR 3251, BP 133,
  91403 Orsay cedex, France and Universit\'e Paris-Sud 11; 
  $^{\text{2}}$Department of Mathematics, Texas A\&M University 3368
  TAMU, College Station, TX 77843-3368, USA;
  $^{\text{3}}$LUTH, Observatoire de
  Paris-Meudon, place Janssen, 92195-Meudon, France;
  $^{\text{3}}$Institut Universitaire de France, 103 Bd Saint-Michel,
75005 Paris, France}

\ead{nore@limsi.fr}

\begin{abstract}
  The purpose of this note is to analyze the long term stability of
  the Navier-Stokes equations supplemented with the Coriolis force and
  the stress-free boundary condition. It is shown that, if the flow
  domain is axisymmetric, spurious stability behaviors can occur
  depending whether the Coriolis force is active or not.
\end{abstract}

\maketitle



\section{Introduction}
The liquid core of the Earth is often modeled as a heated conducting
fluid enclosed between the solid inner core and the mantle.
Numerically simulating the dynamics of the liquid core is difficult in
many respects; one of the difficulties comes from the presence of
viscous layers that develop at the boundaries of the fluid domain, \ie
the so-called inner core boundary (ICB) and core mantle
boundary (CMB). It is a common practice in the geophysics literature
to use stress-free boundary conditions in order to minimize the role
played by the viscous layers. Although this choice of boundary
condition is convenient, it is not clear that it is more physically
justified than using the no-slip condition.  Actually, enforcing
either the no-slip or the stress-free boundary condition may lead to
significantly different results when it comes to simulating the
geodynamo.  For example, Glatzmaier and Roberts~\cite{GR95} and Kuang
and Bloxham~\cite{KB97,KB99} have used the above two different sets of
boundary conditions and have reported numerical buoyancy-driven
dynamos in rapidly rotating spherical shells that differ in some
fundamental aspects, see \eg \cite{Olson97}.  The simulations reported
in~\cite{KB97} use the stress-free condition whereas those reported
in~\cite{GR95} use the no-slip condition.  The dynamo simulated in
\cite{KB97} is composed of an external magnetic field dominated by an
axial dipole component, like that of the Earth, with an intensity
close to the present geomagnetic dipole moment. The external magnetic
field is comparable to that obtained by Glatzmaier and
Roberts~\cite{GR95}, but important differences in the velocity and
magnetic fields between these two dynamos can be observed within the
outer core and the Taylor-Proudman tangent cylinder. (It is known that
rotation of the Earth rigidifies the flow field in the direction
parallel to the rotation axis through a mechanism known as the
Taylor-Proudman effect. This effect makes the imaginary cylinder that
is tangent to the equator of the solid inner core and whose axis is
parallel to the rotation axis of the Earth act like a solid boundary.)
In the dynamo reported in \cite{KB97} the fluid flow is almost
stagnant inside the tangent cylinder and has a strong azimuthal
component outside; the magnetic field is composed of two opposite
toroidal cells and a simple dipolar poloidal structure and is active
throughout the outer core. In the dynamo reported in \cite{GR95} the
fluid flow is composed of an intense polar vortex that is located
inside the tangent cylinder and extends in the two hemispheres; the
toroidal component of the magnetic field is active only inside the
tangent cylinder and is concentrated near the ICB; the poloidal
component has a complicated dipolar structure with extra-closed loops
near the ICB. It is suggested in \cite{Olson97} that the significantly
different structures of the above two dynamos should be attributed to
the nature of the boundary conditions that are imposed at the ICB and
CMB interfaces.

In addition to thermal or compositional convection due to buoyancy,
precession is also believed to be a possible source of energy for the
geodynamo. The precession hypothesis has been formulated for the first
time in~\cite{Bullard49} and experimentally investigated using a water
model in~\cite{Malkus68}. It has since then been actively investigated
from the theoretical, experimental and numerical perspectives.
However, it seems that it is only recently that numerical examples of
precession dynamos have been reported in spheres
\cite{tilgner_precession_2005,tilgner_kinematic_2007}, in spheroidal
cavities \cite{WR09} and in cylinders~\cite{NLGL11}.  Recently, Wu and
Roberts~\cite{WR09} have numerically studied the dynamo effect in a
precessing oblate spheroid. To facilitate their analysis the authors
have split the total velocity field into a basic stationary analytic
(polynomial) solution (the so-called Poincar\'e flow) and a
fluctuating part.  Following ideas of Kerswell and Mason~\cite{MK02},
they have implemented the stress-free boundary condition on the
fluctuating component of the velocity in order to reduce the impact of
the viscous layers at the rigid boundaries.

The purpose of the present paper is to show that the use of the
stress-free boundary condition poses mathematical difficulties.  We
prove for instance that, if the fluid domain is not axisymmetric, the
flow always returns to rest for large times when the stress-free
boundary condition is enforced, but this may not be the case 
when the flow domain is axisymmetric. Various scenarios
can occur depending whether the domain undergoes precession or not.

The note is organized as follows. We analyze the stress-free boundary
condition in general fluid domains in \S\ref{sec:nssf}. We show that
this boundary condition is admissible if and only if the domain is not
axisymmetric (see Proposition~\ref{prop:coerc}). We revisit the same
question in axisymmetric domains that undergo precession in
\S\ref{Sec:Poincare} and \S\ref{Sec:four}. We show in
\S\ref{Sec:Poincare} that the problem exhibits a spurious stability
behavior if the stress-free condition is enforced on the velocity
field minus the Poincar\'e flow (\ie on the perturbation to the
Poincar\'e flow).  We show in \S\ref{Sec:four} that the problem always
returns to rest for large times if the homogeneous stress-free
boundary condition is enforced. The theoretical argumentation
developed in \S\ref{Sec:Poincare} and \S\ref{Sec:four} is numerically
illustrated in \S\ref{Sec:Numerical}. Concluding remarks are reported
in \S\ref{Sec:Conclusions}.

\section{Stress-free boundary condition without
  precession} \label{sec:nssf} The objective of this section is to
investigate the long term stability of the Navier-Stokes equations
equipped with the stress-free boundary condition. The fluid domain is
denoted $\Omega$ and is assumed to be open, bounded and Lipschitz.

\subsection{Position of the problem}
We are interested in the motion of an incompressible fluid in a
container $\Omega$ with boundary $\Gamma$. The container is assumed to
be at rest in a Galilean frame of reference. Denoting $\bu$ the
velocity of the fluid and $p$ the pressure, the fluid motion is
modeled by means of the incompressible Navier-Stokes equations:
\begin{align}
  \partial_t\bu + \bu\ADV\bu -2\nu\DIV\bepsilon(\bu) + \GRAD p  &=
  0, \label{eq:nssym} \\ \DIV \bu &=  0, \label{eq:div}\\ \bu_{|t=0}
   &=  \bu_0, \label{eq:init}
\end{align}
where $\nu$ is the kinematic viscosity, $\bepsilon(\bu) :=
\frac{1}{2}\left(\GRAD\bu+\GRAD\bu\tr\right)$ is the strain rate
tensor, and $\bu_0$ is an initial data in
$\bH:=\{\bv\in\bL^2(\Omega):\ \DIV\bv=0,\ \bv\SCAL
\bn_{_\front}=0\}$. It is a common practice to replace the term
$\DIV\left(\GRAD\bu+\GRAD\bu\tr\right)$ in the momentum equation by
$\LAP\bu$ since $\DIV\GRAD\bu\tr=0$ for incompressible flows. We
nevertheless keep the original form of the viscous stress since we
want to enforce the so-called stress-free boundary condition:
\begin{equation}
\left(\bn\SCAL\bepsilon(\bu)\right)\CROSS\bn_{|\Gamma} = 0,
\label{eq:sfc}
\end{equation}
together with the slip boundary condition: 
\begin{equation} 
\bn\SCAL\bu_{|\Gamma} = 0,\label{eq:nlbc}
\end{equation}
where $\bn$ is the unit outward normal on $\Gamma$. The stress-free
condition means that the tangent component of the stress at the
boundary is zero. We shall see that this boundary condition is
admissible in general for non-axisymmetric domains, but it yields
pathological stability behaviors if the fluid domain is a solid of
revolution.

We are not going to discuss the well-posedness of the above problem in
its full generality since it is still unknown whether the
three-dimensional Navier-Stokes equations are well-posed under the
much simpler no-slip boundary condition.  We nevertheless recognize as
a symptom of pathological stability behavior the fact that there are
solutions to
\eqref{eq:nssym}-\eqref{eq:div}-\eqref{eq:init}-\eqref{eq:sfc}-\eqref{eq:nlbc}
that do not return to rest as $t\to +\infty$ if $\Omega$ is
axisymmetric.

\begin{definition}
We say that $\Omega$ is stress-free admissible if there is a constant
$K>0$, possibly depending on $\Omega$, so that the following holds 
\begin{equation}
  K \int_\Omega \bv^2 \leqslant \int_\Omega
  \bepsilon(\bv){:}\bepsilon(\bv), \qquad \forall\bv \in \Hund,
  \ \bv\SCAL\bn_{|\front} =0
\label{ieq:coerc}
\end{equation}
where $"{:}"$ denotes the tensor double product.
\end{definition}

\begin{proposition} \label{Prop:1}
Assume that $\Omega$ is stress-free admissible, then $\{0\}$ is the global attractor of 
\eqref{eq:nssym}-\eqref{eq:div}-\eqref{eq:init}-\eqref{eq:sfc}-\eqref{eq:nlbc}.
\end{proposition}

\begin{proof} We omit the details concerning the existence of
  Leray-Hopf solutions, which can be constructed using standard
  Galerkin techniques \cite{Lions69}, and we focus only on the aspects
  of the question which are relevant to our discussion. It is clear
  that $0$ is an invariant set of
  \eqref{eq:nssym}-\eqref{eq:div}-\eqref{eq:init}-\eqref{eq:sfc}-\eqref{eq:nlbc}.
  Let $\bB$ be a bounded set in $\bH$ and let $\bu_0\in \bB$.  Let
  $\bu$ be a Leray-Hopf solution corresponding to the initial data
  $\bu_0$ and let $\bv$ be a smooth solenoidal vector field satisfying
  the slip boundary condition. Upon multiplying the momentum equation
  by $\bv$ and integrating over the domain we obtain
  \[ 
  \int_\Omega \partial_t \bu\SCAL\bv + \int_\Omega \bu\ADV\bu\SCAL\bv
  - 2\nu\int_\Omega \DIV\bepsilon(\bu)\SCAL \bv + \int_\Omega \GRAD
  p\SCAL\bv = 0.
  \]
  Solenoidality and the slip boundary condition imply that
  $\int_\Omega \GRAD p\SCAL\bv =-\int_\Omega p\DIV\bv +\int_\Gamma p
  \bv\SCAL \bn=0$. Now, using the decomposition
\[
\bv = (\bn\SCAL\bv)\bn - \bn\CROSS\left(\bn\CROSS\bv\right),
\]
and integrating by parts the viscous term we obtain:
 \begin{align*}
   -\int_\Omega \DIV\bepsilon(\bu)\SCAL \bv &= \int_\Omega
   \bepsilon(\bu){:}\GRAD\bv -
   \int_\Gamma \bn\SCAL\bepsilon(\bu)\SCAL\bv\\
   &= \int_\Omega \bepsilon(\bu){:}\bepsilon(\bv)+ \int_\Gamma
   \left(\bn\SCAL\bepsilon(\bu){\times}\bn\right)\SCAL (\bn\CROSS\bv) =
\int_\Omega \bepsilon(\bu){:}\bepsilon(\bv).
\end{align*}
The transport term and the time derivative are re-written in the
following form
\begin{align*}
  \int_\Omega\bu\ADV\bu\SCAL\bv &=
  \int_\Omega\frac12\DIV(\bu(\bu\SCAL\bv)) +
  \frac12\int_{\Omega}\left(\bu\ADV\bu \SCAL \bv - \bu\ADV\bv \SCAL
    \bu\right) =\frac12\int_{\Omega}\left(\bu\ADV\bu \SCAL \bv -
    \bu\ADV\bv \SCAL \bu\right),
\end{align*}
\[
 \int_\Omega \partial_t \bu\SCAL\bv = \frac12\int_\Omega \partial_t
(\bu\SCAL\bv) +\frac12 \int_\Omega \left(\partial_t\bu\SCAL\bv -
\partial_t \bv \SCAL \bu\right).
\]
We now apply the above identities by replacing $\bv$ by a sequence
$\{\bv_n\}_{n\in \polN}$ that converges in the appropriate norm to
$\bu$. By passing to the limit (we omit the details again), we finally
obtain
\[
\frac12\frac{\diff}{\diff t}\int_\Omega \bu^2 + 2\nu \int_\Omega
\bepsilon(\bu){:}\bepsilon(\bu) \le 0.
\]
Note that equality is lost in the passage to the limit.  Whether
equality holds in general is an open problem which is part of the
Millenium prize.  Then using \eqref{ieq:coerc}, we infer the following
inequality:
\[
\frac 1 2\frac d{dt}\int_\Omega \bu^2 + 2K\nu \int_\Omega \bu^2
\leqslant 0,
\]
which immediately leads to
\[
\|\bu\|_{\bL^2(\Omega)} \leqslant \|\bu_0\|_{\bL^2(\Omega)} \textrm{e}^{-2K\nu t},
\]
thereby proving that $\bu\to0$ as $t\to +\infty$.
\end{proof}
We shall see that the stress-free admissibility condition
\eqref{ieq:coerc} does not hold for axisymmetric fluid domains, which are
common in geoscience.

\subsection{The non-axisymmetric case}
To better understand the stress-free admissibility condition
\eqref{ieq:coerc}, we first prove that it holds if and only if $\Omega$ is not
axisymmetric.

\begin{definition}
  We say that $\Omega$ is axisymmetric (or is solid of revolution) if
  and only if there is a rotation $\bR:\Omega\longrightarrow \Omega$
  which is tangent on $\front$.
\end{definition}
Upon introducing the average operator over $\Omega$, $\langle v
\rangle := \frac{1}{|\Omega|} \int_\Omega v$, where $|\Omega|$ is the
volume of $\Omega$, the following lemma gives a characterization of
non-axisymmetric domains:
\begin{lemma}[Desvillettes-Villani \cite{MR1932965}]  \label{lem:tker}
Assume that the domain $\Omega$ is not a solid of revolution of class $\calC^1$, then
there is $c>0$ so that
\[
c|\Omega| \langle \ROT \bv\rangle^2 \le
\|\bepsilon(\bv)\|_{\bL^2(\Omega)}, \qquad \forall\bv\in \Hund,\
\bv\SCAL\bn_{|\front} =0.
\]
\end{lemma}
We are now in measure to state the main result of this section:
\begin{proposition}\label{prop:coerc}
  Assume that the domain $\Omega$ is of
  class $\calC^1$, then $\Omega$ is stress-free admissible
if and only if $\Omega$ is not a solid of revolution.
\end{proposition}
\begin{proof}
  Let us assume first that $\Omega$ is not a solid of revolution and
  \refp{ieq:coerc} does not hold.  We start from the Korn inequality
  (\cf \eg \cite{bkDuvaut}): there exists a constant $c>0$ such that,
  for all $\bv\in\Hund$,
\begin{equation} 
\|\bv\|_{\Ldeuxd} + \|\GRAD\bv\|_{\Ldeuxd} \le c\left(\|\bv\|_{\Ldeuxd} +
\|\GRAD\bv+\GRAD\bv\tr\|_{\Ldeuxd}\right).
\label{ineq:korn}
\end{equation} 
Since \refp{ieq:coerc} does not hold, for any
$n\in\polN$, one can find $\bu_n\in \Hund$ such that
\[
\bu_n\SCAL\bn_{|\front}=0, \qquad \|\bu_n\|_{\Ldeuxd} = 1, \quad
\textnormal{ and } \quad \|\GRAD\bu_n+\GRAD\bu_n\tr\|_{\Ldeuxd}
\leqslant \frac1n.
\]
The Korn inequality implies that the sequence $\bu_n$ is bounded in
$\Hund$. Since the inclusion $\Hund\subset\Ldeuxt$ is compact, there
exists $\bu\in\Hund$ such that (we keep using $\bu_n$ after
extraction of the converging sub-sequence) $\|\bu_n-\bu\|_{\Ldeuxt}
\rightarrow 0$ and $\bu_n\rightharpoonup \bu$ in $\Hund$. We also have
\[ 
\GRAD\bu_n+\GRAD\bu_n\tr \rightarrow 0 \textnormal{ in } \Ldeuxt
\textnormal{ and } \GRAD\bu_n+\GRAD\bu_n\tr \rightarrow
\GRAD\bu+\GRAD\bu\tr \textnormal{ in } \pmb{\mathcal D}'(\Omega),
\]
which finally gives $\GRAD\bu+\GRAD\bu\tr = 0$ ($\pmb{\mathcal D}(\Omega)$ is the space of smooth
vector-valued functions with compact support in $\Omega$ and
$\pmb{\mathcal D}'(\Omega)$ is the space of vector-valued distributions over
$\Omega$, \ie the linear forms acting on $\pmb{\mathcal D}(\Omega)$.) Applying
the Korn inequality to $\bu-\bu_n$ and using the fact that
\[
\|\bu_n-\bu\|_{\Ldeuxd} +
\|\GRAD\bu_n+\GRAD\bu_n\tr-\GRAD\bu-\GRAD\bu\tr\|_{\Ldeuxd} \rightarrow
0,
\]
we infer that $\|\bu_n-\bu\|_{\Hund}\rightarrow 0$.  This allows us to
pass to the limit on the boundary condition $\bu\SCAL\bn_{|\front}=0$.
The condition $\bepsilon(\bu)=0$ implies that there are two vectors
$\bt\in \Real^3$, $\bomega\in \Real^3$ so that $\bu = \bt +
\bomega{\times}\bx$. This means that $\ROT \bu = \langle\ROT \bu
\rangle = \bomega$. Using Lemma \refp{lem:tker}, we conclude that
$\b\omega = 0$, which means that $\bu=\bt$. The boundary condition
$\bu\SCAL\bn_{|\front}=0$ implies $\bt=0$; this in turn means $\bu=0$,
which is impossible because $\|\bu\|_{\Ldeuxt} = 1$. In conclusion,
\refp{ieq:coerc} holds.

Let us assume now that $\Omega$ is axisymmetric. This means that
there is a rotation $\bR: \Omega \longrightarrow \Omega$ which is
tangent on $\front$. Let us assume that the rotation axis is parallel
to $\be_z$ and the coordinate origin is located on this axis. Then
$\bR(\bx)= \omega \be_z {\times}\bx$ and clearly $\bR\in \Hund$,
$\bR(\bx)\SCAL \bn(\bx)_{|\front} =0$, $\|\bR\|_{\Ldeuxd}\not=0$ but
\refp{ieq:coerc} does not hold since $\bepsilon(\bR)=0$.
\end{proof}

\subsection{The Axisymmetry curse}
Let us assume that $\Omega$ is axisymmetric. We are going to show the
following statement in this section.
\begin{claim}
The zero velocity field, $0$, is in the global attractor of
\eqref{eq:nssym}-\eqref{eq:div}-\eqref{eq:init}-\eqref{eq:sfc}-\eqref{eq:nlbc},
but the rest state, $\{0\}$, is not an attractor. There are initial
data that create flows that never return to rest. In particular, if
the initial data is a solid rotation, the flow will rotate for ever
without losing energy.
\end{claim}

Recall that it can be shown that $\Omega$ is axisymmetric if and only
if  $\Omega$ is either a sphere (and all the directions are symmetry
axes) or $\Omega$ has a unique symmetry axis.  Without a loss of
generality, we assume $Oz$ is the only symmetry axis of
$\Omega$. Recall that all the solid rotations about $Oz$ can be
written as follows $\bx\longmapsto \omega \be_z\CROSS\bx$, $\omega\in
\Real$, where $\bx$ is the position vector.  We introduce the
following space
\begin{equation}
\calR :=
\text{span}\left\{\be_z\CROSS\bx\right\}
\end{equation}
and its orthogonal in $\bL^2(\Omega)$, say $\calR^\perp$.

\begin{lemma} \label{lem:coerc_uperp} Let $\Omega$ be an open, bounded, connected,
  domain of class $\calC^1$ with unique symmetry axis $Oz$.  There
  exists $K>0$ such that, for every $\bv\in\calR^\perp\cap
  \bH^1(\Omega)$ with $\bv\SCAL\bn = 0$
\[
K\|\bv\|_{\bL^2(\Omega)}^2 \le \int_\Omega |\bepsilon(\bv)|^2,
\]
\end{lemma}
where we denote $|\bepsilon(\bv)|^2:=\bepsilon(\bv){:}\bepsilon(\bv)$.
\begin{proof}
  The proof is similar to that of Proposition~\ref{prop:coerc}. By
  contradiction, we consider a sequence $\bv_n\in\calR^\perp\cap
  \bH^1(\Omega)$ with vanishing normal component such that
\[
\|\bv_n\|_{\bL^2} = 1\textnormal{ and }
\|\bepsilon(\bv_n)\tr\|_{\bL^2}\leqslant\frac1n.
\]
Using Korn inequality, we can prove that (up to extraction) $\bv_n$
converges in $\bH^1(\Omega)$, and the limit $\bv$ satisfies
\[
\bv\in\calR^\perp,\; \bv\SCAL\bn = 0\textnormal{ and
} \bepsilon(\bv)= 0.
\]
This implies that $\bv$ is the sum of a translation plus a solid
rotation. But $\Omega$ being bounded the
translation is zero and $\bv$ is a solid rotation about $Oz$-axis
(recall that $\beta\not=0$), \ie $\bv\in \calR\cap\calR^\perp=\{0\}$,
which contradicts $\|\bv\|_{\bL^2(\Omega)} = 1$.
\end{proof}

We claim that the Navier-Stokes problem
\eqref{eq:nssym}-\eqref{eq:div} equipped with boundary conditions
\eqref{eq:sfc}-\eqref{eq:nlbc} has spurious stability properties due
to the following proposition.
 \begin{proposition}
   (i) $\calR$ is the global attractor of
   \eqref{eq:nssym}-\eqref{eq:div}-\eqref{eq:init}-\eqref{eq:sfc}-\eqref{eq:nlbc}. (ii)
   No element in $\calR$ is an attractor.
\label{prop:NS_illposed}
\end{proposition}

\begin{proof} (i)   Let $\bu\in
  L^2((0,+\infty);\bL^2(\Omega))\cap L^\infty((0,+\infty);\bH^1(\Omega))$
  be a Leray-Hopf solution of \eqref{eq:nssym}--\eqref{eq:nlbc} and
  consider the following decomposition:
\[
\bu(t) = \bu^\perp(t)+\lambda(t)\be_z\CROSS\bx,\textnormal{ where }
\bu^\perp(t)\in\calR^\perp,\ \lambda(t)\in\mathbb R,\ \forall t\in [0,+\infty).
\]
$\bu$ being a Leray-Hopf solution implies that
\[
\|\bu^\perp(t)\|_{\bL^2(\Omega)}^2 + \lambda(t)^2
\|\be_z\CROSS\bx\|_{\bL^2(\Omega)}^2 + 4\nu \int_0^t \int_\Omega |\bepsilon(\bu)|^2 \le\|\bu_0\|_{\bL^2(\Omega)}^2.
\]
which owing to Lemma~\ref{lem:coerc_uperp} implies
\[
\|\bu^\perp(t)\|_{\bL^2(\Omega)}^2 + \lambda(t)^2
\|\be_z\CROSS\bx\|_{\bL^2(\Omega)}^2 + 4\nu K \int_0^t
\|\bu^\perp\|_{\bL^2(\Omega)}^2\diff\tau \le\|\bu_0\|_{\bL^2(\Omega)}^2.
\]
Using the Gronwall-Bellmann inequality, we infer that
$\|\bu^\perp(t)\|_{\bL^2(\Omega)} \le \|\bu_0\|_{\bL^2(\Omega)}
\textrm{e}^{-2\nu K t}$.  Invoking Lemma~\ref{lemme:moment} we infer
that $\frac{\diff \lambda(t)}{\diff t}=0$, implying that
$\lambda(t)=\lambda(0)$. In conclusion
\[
\|\bu(t) -\lambda_0 \be_z\CROSS\bx\|_{\bL^2(\Omega)} = \|\bu^\perp(t)\|_{\bL^2(\Omega)} \le \|\bu_0\|_{\bL^2(\Omega)}
\textrm{e}^{-2\nu K t}.
\]
This implies that the global attractor, say $\calA$, is such that
$\calA\subset \calR$, but since $\lambda_0$ spans $\Real$, we
conclude that $\calA=\calR$.

  (ii) Let us consider the solid rotation field $\bu= \omega
  \be_z\CROSS\bx\in \calR$. It is clear that $\bu$ is invariant, \ie
  is a steady-state solution. Let $\bB(\bu,\rho)\in \bH$ be the ball
  centered at $\bu$ of arbitrary radius $\rho>0$. Let
  $\bv=\mu\be_z\CROSS\bx\in \calR$, $\mu\not=0$, be another solid
  rotation and assume that $\mu$ is small enough so that
  $\bu+\bv\in\bB(\bu,\rho)$.  Let us observe that
  \[
  (\bu+\bv)\ADV(\bu+\bv) = 2(\bu+\bv)\SCAL \bepsilon(\bu+\bv)-
  (\bu+\bv)\SCAL(\GRAD(\bu+\bv))^T=-\frac12 \GRAD|\bu+\bv|^2,
  \] since $\bu+\bv$ is a solid rotation; moreover, $\bu+\bv$
  satisfies \refp{eq:div}, \refp{eq:nlbc} and
  $\bepsilon(\bu+\bv)=0$. The property $\bepsilon(\bu)=0$ implies
  \eqref{eq:sfc} and $\DIV(\bepsilon(\bu))=0$. Upon setting $p=\frac12
  |\bu+\bv|^2$ we conclude that $\bu+\bv$ solves
  \eqref{eq:nssym}. This proves that $\bu+\bv$ is invariant (\ie a
  steady-state solution).  In other words $\bu+\bv$ does not converge
  to $\bu$, no matter how small $\rho$ is, thereby proving that the
  set $\{\bu\}$ is not an attractor, no matter how large $\nu$ is.
\end{proof}

\subsection{An admissible stress-free-like boundary condition} 
\label{Sec:admissible}
The principal motivation to consider the so-called stress-free
boundary condition is that it minimizes viscous layers and is thus
less computationally demanding than the no-slip boundary condition.
We have seen above that this boundary condition unfortunately leads to
pathological stability properties when the computational domain is
axisymmetric.  A possible remedy to this problem is to consider the
following non-symmetric boundary condition:
\begin{equation}
\left(\bn\SCAL\GRAD\bu\right)\CROSS\bn_{|\front} = 0.
\label{eq:nsfbc}
\end{equation}
The tangent components of the normal derivative of the velocity field
are zero. The physical interpretation of this condition is definitely
less appealing than that of the stress-free boundary condition, but
once one realizes that the stress-free boundary condition is ad hoc,
one comes to think that \eqref{eq:nsfbc} is not more ad hoc than the
stress-free condition.  The main advantage we see in \eqref{eq:nsfbc}
over the stress-free condition is that it yields standard
stability properties, \ie $\{0\}$ is the global attractor when there
is no forcing.
\begin{lemma} \label{Lem:epsilon_gradient} The following holds for all
  smooth solenoidal vector field $\bu$ that satisfies
  $\left(\bn\SCAL\GRAD\bu\right)\CROSS\bn_{|\front}=0$:
\begin{equation}
\int_\Omega -\DIV(\bepsilon(\bu))\SCAL\bv =\tfrac12 \int_\Omega
\GRAD\bu{:}\GRAD\bv, \qquad \forall \bv\in \Hund,\ \bv\SCAL\bn_{|\front}=0.
\end{equation}
\end{lemma}

\begin{proof}
Upon observing that $\DIV(\bepsilon(\bu))=\tfrac12 \DIV (\GRAD \bu)$ since $\bu$ is
solenoidal, we infer that
\begin{align*}
\int_\Omega -\DIV(\bepsilon(\bu))\SCAL\bv &= 
\int_\Omega -\tfrac12 \DIV (\GRAD \bu)\SCAL\bv = 
\tfrac12 \int_\Omega  \GRAD \bu{:} \GRAD\bv  - \tfrac12 \int_\front (\bn\SCAL\GRAD \bu)\SCAL\bv \\
& = \tfrac12 \int_\Omega  \GRAD \bu{:} \GRAD\bv 
- \tfrac12 \int_\front (\bn\SCAL\GRAD \bu)\SCAL((\bn\SCAL\bv)\bn) 
+  \tfrac12 \int_\front ((\bn\SCAL\GRAD \bu)\CROSS \bn) \SCAL\left(\bn\CROSS\bv\right)\\
& = \tfrac12 \int_\Omega  \GRAD \bu{:} \GRAD\bv, 
\end{align*}
where we used again the decomposition $\bv_{|\front} = (\bn\SCAL\bv)\bn -
\bn\CROSS\left(\bn\CROSS\bv\right)$.
\end{proof}

\begin{proposition}
  Assume that $\Omega$ is an open, connected, bounded Lipschitz domain, then $\{0\}$
  is the global attractor of
  \eqref{eq:nssym}-\eqref{eq:div}-\eqref{eq:init}-\eqref{eq:nsfbc}-\eqref{eq:nlbc}.
\end{proposition}
  
\begin{proof}
  Repeat the argument in the proof of Proposition~\ref{Prop:1} using
  Lemma~\ref{Lem:epsilon_gradient} together with the following
  Poincar\'e-like inequality
\[
K \int_\Omega \bv^2 \le \int_\Omega |\GRAD\bv|^2,\qquad \forall
\bv\in\Hund,\ \bv\SCAL\bn_{|\front}=0,
\]
which can be shown to hold by proceeding as in the proof of
Proposition~\ref{prop:coerc}.
\end{proof}

\section{Precession driven flow with Poincar\'e stress}
\label{Sec:Poincare}
If the fluid domain is a spheroid that undergoes precession, the
steady state Navier-Stokes equations with the slip condition admit a
so-called Poincar\'e solution. We show in this section that,
independently of the value of the viscosity, the Poincar\'e solution
is not an attractor of the problem if the tangential stress at the
boundary is enforced to be equal to that of the steady-state
Poincar\'e solution.

\subsection{Geometry and equations} \label{Sec:Poincare_setting}
The container is an ellipsoid of revolution of center $O$ and symmetry
axis $Oz$. The unit vector along the $Oz$-axis is $\be_z$. The unit
vectors along the other two orthogonal axes $Ox$ and $Oy$ are $\be_x$
and $\be_y$, respectively.  The surface of the ellipsoid is defined by
the equation
\begin{equation} 
  x^2 + y^2 + (1+\beta)z^2 = 1, \label{def:omega}
\end{equation}
where $\beta>-1$ and $\beta\not=0$.  We assume that the container
rotates about the $Oz$-axis with angular velocity $\be_z$ and that
this frame slowly precesses about the $Ox$-axis with angular velocity
$\varepsilon\be_x$ (this particular precession angle is investigated
in \cite{WR09}).  The non-dimensional Navier-Stokes equations
describing the motion of the fluid in the non-inertial precessing
frame of reference $(O,\be_x,\be_y,\be_z)$ are written as follows:
\begin{align}
\partial_t\bu + \bu\ADV\bu -2\nu\DIV\bepsilon(\bu) + 2\varepsilon\be_x\CROSS\bu + \GRAD p
& =  0, \label{eq:nsp} \\ \DIV \bu & =  0, \label{eq:divp} \\
\bu_{|t=0} &=\bu_0.
\end{align}
We additionally enforce the slip boundary condition,
\begin{equation}
\bu\SCAL\bn_{|\front}=0. \label{eq:slip}
\end{equation}

The system \eqref{eq:nsp}-\eqref{eq:divp}-\eqref{eq:slip} is known to
admit a steady solution called the Poincar\'e flow (see \eg
\cite{WR09}); its expression is:
\begin{equation} \bu_P = -y\be_x +
\left(x-\frac{2\varepsilon}{\beta}(1+\beta)z\right)\be_y +
\frac{2\varepsilon}{\beta}y\be_z.
\end{equation} 
Similarly to \cite{WR09} we consider the problem
\eqref{eq:nsp}-\eqref{eq:divp}-\eqref{eq:slip} equipped with the
additional non-homogeneous boundary condition
\begin{equation}
\left(\bn\SCAL\bepsilon(\bu)\right)\CROSS\bn_{|\front} =
\left(\bn\SCAL\bepsilon(\bu_P)\right)\CROSS\bn_{|\front}.  \label{eq:Poincare_stress}
\end{equation}
That is, we want the tangential component of the normal stress to be
equal to that of the Poincar\'e solution. As mentioned in \cite{WR09},
it is clear that
\begin{claim}[See \cite{WR09}]
  $\bu_P$ is a steady state solution of
  \eqref{eq:nsp}-\eqref{eq:divp}-\eqref{eq:slip}-\eqref{eq:Poincare_stress}.
\end{claim}

\subsection{Long term stability}
The question that we now want to investigate is whether there is a
threshold on $\nu$ beyond which $\bu_P$ is a stable solution as $t\to
+\infty$; \ie does the flow return to $\bu_P$ independently of the
initial data as $t \to +\infty$ if $\nu$ is large enough?  We show in
this section that the answer to this question is no, the fundamental
reason being that solid rotations cannot be dampened by viscous
dissipation, no matter how large $\nu$ is.
\begin{proposition} For all $\nu>0$, $\{\bu_P\}$ is not an attractor
  of the Navier-Stokes problem \eqref{eq:nsp}-\eqref{eq:divp} equipped
  with the boundary conditions
  \eqref{eq:slip}-\eqref{eq:Poincare_stress}.
\label{prop:NSSF}
\end{proposition}

\begin{proof} Let $\rho>0$ be an arbitrary positive number. Let
  $\bB(\bu_P,\rho)\subset \bH$ be a ball of radius $\rho$ centered at
  $\bu_P$. Let $\bw = \omega\be_z\CROSS\br$ is a solid rotation about
  the $Oz$-axis, and assume that $\omega\not=0$ is small enough so
  that $\bu_P+\bw\in \bB(\bu_P,\rho)$.  Let us prove that $\bu_P +
  \bw$ is a steady state solution of
  \eqref{eq:nsp}-\eqref{eq:divp}-\eqref{eq:slip}-\eqref{eq:Poincare_stress}.
  Owing to $\bepsilon(\bw)=0$, $\bw\SCAL\bn_{|\front}=0$,
  $\DIV\bw=0$, it is clear that $\bu_P+\bw$ is solenoidal and satisfies the
  boundary conditions \eqref{eq:slip}-\eqref{eq:Poincare_stress}.  Let
  us now show that it is possible to find a pressure field so that the
  steady state momentum equation holds.

  Let us first prove that $\bu_P\ADV\bw + \bw\ADV\bu_P +
  2\varepsilon\be_x\CROSS\bw$ is a gradient. A straightforward computation
  gives:
\[
\bu_P\ADV\bw = \omega
\left(\begin{matrix}\frac{2\varepsilon}{\beta}(1+\beta)z- x\\ -y \\
    0 \end{matrix}\right)\!\!, \quad \bw\ADV\bu_P = \omega
\left(\begin{matrix} - x \\ -y
    \\\frac{2\varepsilon}{\beta}x \end{matrix}\right)\!\!, \quad
2\varepsilon\be_x\CROSS\bw = \omega \left(\begin{matrix} 0 \\ 0 \\2\varepsilon
    x \end{matrix}\right)\!\!,
\]
so that
\[
\bu_P\ADV\bw + \bw\ADV\bu_P + 2\varepsilon\be_x\CROSS\bw =
-\GRAD\left(\omega(x^2+y^2)\right) +
\GRAD\left(\frac{2\varepsilon\omega}{\beta}(1+\beta)xz\right).
\]
Let use then define $q(\bx):= -\omega(x^2+y^2) +
\frac{2\varepsilon\omega}{\beta}(1+\beta)xz$. Observe that we can
define the pressure field $r(\bx)$ so that $\GRAD r := -\bu_P\ADV\bu_P
-2\varepsilon \be_x{\times}\bu_P$, since $\bu_P$ solves
\eqref{eq:nsp}. Let us finally observe that
$\bw\ADV\bw=-\frac12\GRAD|\bw|^2$. Then we conclude that $\bu_P+\bw$
solves \eqref{eq:nsp} with $p = q+r-\frac12|\bw|^2$. In particular if
we set $\bu_0=\bu_P+\bw$, then $\bu_P+\bw$ remains a solution forever,
\ie the solution does not converge to $\bu_P$ as $t\to +\infty$, no
matter how small $\rho$ is and no matter how large $\nu$ is.
\end{proof}

\subsection{Angular momentum balance} Let us now mention a result on
the balance of the angular momentum. Let us assume that $\bu$ solves
\eqref{eq:nsp}-\eqref{eq:divp} with the boundary conditions
\begin{align}
  \bn\SCAL\bu & =  0 && \textnormal{ on }\Gamma, \label{eq:boundary1}\\
  \left(\bn\SCAL\bepsilon(\bu)\right)\CROSS\bn & = \bg \CROSS\bn&&
  \textnormal{ on }\Gamma, \label{eq:boundary2}
\end{align}
where the field $\bg$ is a boundary data. Let us now define the
angular momentum
\begin{equation}
  \bM:=\int_\Omega \bx\CROSS\bu.
\end{equation}

\begin{lemma} \label{lemme:moment} Denoting by $M_z$ and $M_y$ the $z$-
  and $y$-component of $\bM$, respectively, all the weak solutions of
  \eqref{eq:nsp}-\eqref{eq:divp}-\eqref{eq:boundary1}-\eqref{eq:boundary2}
  satisfy
\begin{equation}
  \partial_t M_z + \varepsilon M_y = -\int_{\partial\Omega} \nu (\bg\CROSS\bn) \SCAL((\be_z\CROSS\bx)\CROSS\bn),
  \qquad \text{a.e. } t\in (0,+\infty).
\label{eq:moment}
\end{equation}
\end{lemma}
\begin{proof}
  Observing that $M_z = \int_\Omega (\be_z\CROSS\bx)\SCAL \bu$, we
  multiply \refp{eq:nsp} by $\be_z\CROSS\bx$ and integrate over
  $\Omega$. Using the divergence free condition together with
  \eqref{eq:boundary1} and integrating by parts, we infer that
\begin{align*}
  \int_\Omega (\be_z\CROSS\bx)\SCAL (\bu\ADV\bu) & = \int_\Omega
  \DIV(\bu\otimes\bu)\SCAL(\be_z\CROSS\bx) = \int_{\partial\Omega}
  \left(\bu\SCAL\bn\right)\left(\bu\SCAL\left(\be_z\CROSS\bx\right)\right)=0,
\end{align*}
where we used that $(\bu\otimes\bu){:}\GRAD(\be_z\CROSS\bx)=0$ since
the matrix $\bu\otimes\bu$ is symmetric and $\GRAD(\be_z\CROSS\bx)$ is
anti-symmetric. The same argument applies to the viscous term
\[
  \int_\Omega (\be_z\CROSS\bx)\SCAL \nu\DIV(\bepsilon(\bu))
  =  \int_{\partial\Omega} \nu (\bepsilon(\bu)\SCAL\bn)\SCAL(\be_z\CROSS\bx)
  =  \int_{\partial\Omega} \nu (\bg\CROSS\bn) \SCAL((\be_z\CROSS\bx)\CROSS\bn), 
\]
where we used $\be_z\CROSS\bx =(\be_z\CROSS\bx)\CROSS\bn$ since
$(\be_z\CROSS\bx)\SCAL\bn_{|\front}=0$. The same argument applies again
for the pressure term since $\GRAD p =\DIV (p I)$ where $I$ is the
identity matrix.
\[
  \int_\Omega (\be_z\CROSS\bx)\SCAL \GRAD p  = \int_{\partial\Omega}
  p(\be_z\CROSS\bx)\SCAL\bn =0.
\]
We now deal with the Coriolis term by applying Lemma~\ref{Lem:My}: 
\[
\int_\Omega (\be_z\CROSS\bx)\SCAL(\be_x\CROSS\bu) =  \frac12\int_\Omega
\be_y\SCAL (\bx\CROSS\bu)=\frac12 M_y.
\]
The conclusion follows readily.
\end{proof}

\begin{lemma} 
\label{Lem:My} 
Let $\bv\in\bL^1(\Omega)$ be an integrable vector field 
such that $\DIV\bv=0$ and $\bv\SCAL\bn_{|\front=0}$, then
\begin{equation}
\int_\Omega\be_y\SCAL (\bx\CROSS\bv) =2\int_\Omega (\be_z\CROSS\bx)\SCAL(\be_x\CROSS\bv).
\end{equation}
\end{lemma}
\begin{proof}
Let us first observe that $ \int_\Omega
(\be_z\CROSS\bx)\SCAL(\be_x\CROSS\bv) = - \int_\Omega x u_z$.
Noticing that $\int_\Omega x v_z + z v_x = \int_\Omega
\bv\SCAL\GRAD(zx)=0$ since $\DIV\bv=0$ and $\bv\SCAL\bn_{|\front}=0$,
we infer that
\[
\int_\Omega (\be_z\CROSS\bx)\SCAL(\be_x\CROSS\bv) = - \int_\Omega x
u_z = \frac12\int_\Omega zv_x-xv_z = \frac12\int_\Omega
\be_y\SCAL (\bx\CROSS\bv),
\]
which concludes the proof.
\end{proof}

\begin{remark}
  If we choose $\bg = \bepsilon(\bu_P)\SCAL\bn$ like in
  \eqref{eq:Poincare_stress}, then $ -\int_{\partial\Omega} \nu
  (\bg\CROSS\bn) \SCAL((\be_z\CROSS\bx)\CROSS\bn)$ is equal to
  $-\int_\Omega (\be_z\CROSS\bx)\SCAL \nu\DIV(\bepsilon(\bu_P))=0$ and
  the balance equation of the angular momentum in the $z$ direction
  simplifies to $\partial_t M_z + \varepsilon M_y=0$.
\end{remark}

\begin{remark}
  Note that \eqref{eq:moment} is just a consequence of
  \eqref{eq:nsp}-\eqref{eq:divp}-\eqref{eq:boundary1}-\eqref{eq:boundary2}.
  This balance holds whether the long term stability of
  \eqref{eq:nsp}-\eqref{eq:divp}-\eqref{eq:boundary1}-\eqref{eq:boundary2}
  is spurious or not. It is false to consider that \eqref{eq:moment}
  is an additional equation that fixes the long term stability
  behavior of
  \eqref{eq:nsp}-\eqref{eq:divp}-\eqref{eq:slip}-\eqref{eq:Poincare_stress}.
\end{remark}

\section{Precession driven flow with stress-free boundary conditions}
\label{Sec:four}
We show in this section that if we enforce
$\bepsilon(\bu)\SCAL\bn_{|\front}=0$, instead of enforcing
$\bepsilon(\bu)\SCAL\bn_{|\front}=\bepsilon(\bu_P)\SCAL\bn_{|\front}$
in \eqref{eq:nsp}-\eqref{eq:divp}-\eqref{eq:boundary1}, then $0$
becomes the unique stable solution as $t\to +\infty$, \ie $\{0\}$ is
the global attractor.

\subsection{Long time stability}
The setting of the problem is the same as in
Section~\ref{Sec:Poincare_setting} except that we enforce the
tangential component of the normal stress to be zero at the boundary.
\begin{align}
  \partial_t\bu + \bu\SCAL\GRAD\bu - 2\nu\DIV\bepsilon(\bu) +
  2\varepsilon\be_x\CROSS\bu + \GRAD p &= 0 &\textnormal{ in
  }\Omega \label{eq:pdfsf1} \\ \DIV\bu &= 0
  &\textnormal{ in }\Omega \label{eq:pdfsf2} \\ \bn\SCAL\bu &=0 &\textnormal{ on }\Gamma  \label{eq:pdfsf3}\\
  \left(\bn\SCAL\bepsilon(\bu)\right)\CROSS\bn&=
  0 &\textnormal{ on }\Gamma \label{eq:pdfsf4} \\
  \bu_{|t=0} &= \bu_0 &\textnormal{ in }\Omega. \label{eq:pdfsf5} 
\end{align}
The result that we want to emphasize is that contrary to what we
observed in Section~\ref{Sec:Poincare}, $0$ becomes the unique stable
solution of \eqref{eq:pdfsf1}--\eqref{eq:pdfsf5} as $t\to +\infty$.
The main result that we want to prove here is that any solution of the
system \eqref{eq:pdfsf1}-\eqref{eq:pdfsf4} returns to rest as $t\to
+\infty$. The key argument is that solid rotations about the $Oz$ axis
are not stationary solutions of \eqref{eq:pdfsf1}. This fact has been
mentioned in \cite{WR09} without proof.
\begin{theorem} \label{thm:cvg} $\{0\}$ is the global attractor of
\eqref{eq:pdfsf1}--\eqref{eq:pdfsf5}.
\end{theorem}
\begin{proof} Let us start by observing that $\{0\}$ is indeed an
  invariant set of \eqref{eq:pdfsf1}--\eqref{eq:pdfsf5}.  Let
  $\bB(0,\rho)$ be the unit ball in $\bH$ centered at $0$ and of
  radius $\rho$. Let $\bu_0\in \bB(0,\rho)$ and let $\bu\in
  L^2((0,+\infty);\bL^2(\Omega))\cap
  L^\infty((0,+\infty);\bH^1(\Omega))$ be a Leray-Hopf solution of
  \eqref{eq:pdfsf1}--\eqref{eq:pdfsf5} and consider the following
  decomposition:
\[
\bu(t) = \bu^\perp(t)+\lambda(t)\be_z\CROSS\bx,\textnormal{ where }
\bu^\perp(t)\in\calR^\perp,\ \lambda(t)\in\mathbb R,\ \forall t\in [0,+\infty).
\]
Lemma~\ref{lem:coerc_uperp} together with $\bu$ being a Leray-Hopf
solution implies that
\[
\|\bu^\perp(t)\|_{\bL^2(\Omega)}^2 + \gamma \lambda(t)^2
 + 4\nu K \int_0^t
\|\bu^\perp(\tau)\|_{\bL^2(\Omega)}^2\diff\tau \le
  \|\bu_0\|_{\bL^2(\Omega)}^2.
\]
where $\gamma=\|\be_z\CROSS\bx\|_{\bL^2}^2$.
Using the Gronwall-Bellmann inequality, we infer that
$\|\bu^\perp(t)\|_{\bL^2(\Omega)} \le \|\bu_0\|_{\bL^2(\Omega)}
\textrm{e}^{-2\nu K t}$.

Let $t_2 > t_1$ in $(0,+\infty)$, then \eqref{eq:moment} means that
\begin{align*}
  (\lambda(t_2)-\lambda(t_1))\gamma =
  -\varepsilon \int_{t_1}^{t_2}\int_\Omega
  \be_y\SCAL(\bx\CROSS(\lambda(\tau)(\be_z\CROSS\bx) + \bu^\perp)).
\end{align*}
But Lemma~\ref{Lem:My} implying that
\begin{align*}
\int_\Omega \be_y\SCAL(\bx\CROSS(\lambda(\tau)(\be_z\CROSS\bx)))
&=\lambda(\tau)\int_\Omega \be_y\SCAL(\bx\CROSS(\be_z\CROSS\bx)) \\
&= 2 \lambda(\tau)\int_\Omega (\be_z\CROSS\bx)\SCAL(\be_x\CROSS(\be_z\CROSS\bx))
=0,
\end{align*}
we finally infer that
\begin{align*}
  |\lambda(t_2)-\lambda(t_1)|
\le \gamma^{-1} \varepsilon \int_{t_1}^{t_2}\int_\Omega
  |\be_y\SCAL(\bx\CROSS \bu^\perp)|\le
  c\, (\textrm{e}^{-2\nu Kt_1} - \textrm{e}^{-2\nu K t_2}),
\end{align*}
where $c$ is a generic constant that depends on $\Omega$, $\nu$, and
$\rho$ and may vary at each occurrence from now on. Note in passing
that this also proves that $\lambda(t)$ converges to a real number
$\lambda_\infty$ as $t\to +\infty$, and $|\lambda_\infty
-\lambda(t)|\le c\, \textrm{e}^{-2\nu K t}$.

Let us take $\bvarphi\in\pmb{\mathcal D}(\Omega)$ independent of time
and divergence-free. Since $\bu$ is a Leray-Hopf solution (recall
$t\longmapsto \bu(t)$ is continuous in the $\bL^2$-weak topology) we
have
\begin{align*}
  0 &=\int_\Omega (\bu(t_2,\bx)-\bu(t_1,\bx))
  \SCAL\bvarphi(\bx)\diff\bx + \int_{t_1}^{t_2} \int_\Omega
  2\varepsilon\bu(\tau,\bx)\SCAL(\bvarphi(\bx)\CROSS\be_x)\diff\bx\diff\tau
  \\ &- \int_{t_1}^{t_2} \int_\Omega
  2\nu\bu(\tau,\bx)\SCAL\DIV\bepsilon(\bvarphi)\diff\bx\diff\tau -
  \int_{t_1}^{t_2} \int_\Omega
  (\bu(\bx)\otimes\bu(\bx)){:}\GRAD\bvarphi(\bx)\textnormal{d}\bx\diff\tau.
\end{align*}
Let us now set $t_2=t_1+1$. Upon observing that $\int_\Omega
((\be_z\CROSS\bx)\otimes(\be_z\CROSS\bx)){:}\GRAD\bvarphi(\bx)\textnormal{d}\bx=0$
and $\int_\Omega (\be_z\CROSS\bx)\SCAL\DIV\bepsilon(\bvarphi)\diff\bx=0$.
This implies that there is a constant $c(\bvarphi)\ge 0$ so that
\begin{align*}
  2 \varepsilon\left|
  \int_{t_1}^{t_2}\lambda(\tau)\diff\tau \int_\Omega
  (\be_z\CROSS\bx)\SCAL(\bvarphi(\bx)\CROSS\be_x)\diff\bx \right|
& \le  \left| (\lambda(t_2)-\lambda(t_1))\int_\Omega
  (\be_z\CROSS\bx)\SCAL\bvarphi(\bx)\diff\bx\right| \\
& + c(\bvarphi) \textrm{e}^{-2\nu K t_1}.
\end{align*}
Let us choose $\bvarphi$ so that $2 \varepsilon\int_\Omega
(\be_z\CROSS\bx)\SCAL(\bvarphi(\bx)\CROSS\be_x)\diff\bx =1$. The
above estimate implies that
\[
\left|\int_{t_1}^{t_1+1}\lambda(\tau)\diff\tau\right|
\le c(\bvarphi) \textrm{e}^{-2\nu K t}.
\]
This in turn implies that $ \lambda_\infty =\lim_{t_1\to \infty}
\int_{t_1}^{t_1+1}\lambda(\tau)\diff \tau =0$, which means
$\lambda_\infty =0$.  In conclusion 
\begin{equation} \label{Energy_precession}
\lim_{t\to +\infty}
\|\bu(t)\|_{\bL^2(\Omega)} \le c(\bvarphi)\, \lim_{t\to +\infty}
\textrm{e}^{-2\nu K t} = 0,
\end{equation} 
which concludes the proof.
\end{proof}

\section{Numerical illustrations}\label{Sec:Numerical}
To illustrate the above mathematical results, we have performed two
series of numerical simulations similar to those presented
in~\cite{WR09}.
The authors study therein the dynamo action in an
oblate spheroid defined by equation~\eqref{def:omega} with
$\beta=0.5625$ (this corresponds to the value $b=0.8$ for the
semi-minor axis used in~\cite{WR09}, $b := (1+\beta)^{-\frac12}$). This spheroid
rotates about the $Oz$-axis and precesses about the $Ox$-axis. Two sets
of boundary conditions are considered: either the homogeneous
stress-free boundary or the Poincar\'e stress condition is enforced.
Simulations are carried out using a mixed Fourier decomposition and
finite element code described in details in~\cite{GLLNR11}. 

The first simulation solves the
equations~\eqref{eq:pdfsf1}-\eqref{eq:pdfsf2}-\eqref{eq:pdfsf3}-\eqref{eq:pdfsf4}
with the initial data $\bu_{|t=0} = 0.1(-y \be_x + x \be_y)$. The
precession rate is $\epsilon=0.25$ and the reciprocal of the viscosity
is $1/\nu = 0.024$. Figure~\ref{fig:Re42} shows the time derivative of
the total energy $E_K= \frac12 \|\bu\|_{\bL^2}^2$ in the precessing
frame. Note that $\partial_t E_K$ is always negative, establishing
that $E_K$ is a decreasing function. This graph is in excellent
agreement with figure 1 of~\cite{WR09}. It also shows that $\bu
\rightarrow 0$ as $t \rightarrow \infty$ in agreement with
\eqref{Energy_precession} (\ie $\{0\}$ is indeed the attractor) .
\begin{figure}[h]
\centerline{
\includegraphics[width=0.5\textwidth]{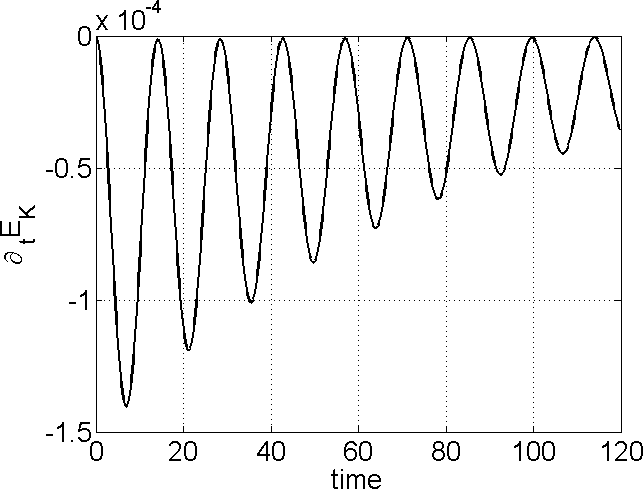}
}
\caption{Time evolution of $\partial_t E_K$ of the solution of
  equations
  \eqref{eq:pdfsf1}-\eqref{eq:pdfsf2}-\eqref{eq:pdfsf3}-\eqref{eq:pdfsf4}
  for $\beta=0.5625,\, \epsilon=0.25$ and $1/\nu = 0.024$.}
\label{fig:Re42}
\end{figure}

The second series of simulations solves
equations~\eqref{eq:nsp}-\eqref{eq:divp}-\eqref{eq:slip}-\eqref{eq:Poincare_stress}
with an {\it ad hoc} initial condition which is a very small
perturbation of the Poincar\'e solution.  The parameters are
$\epsilon=0.25$ and $1/\nu = 0.00375$.
Figure~\ref{fig:tot_Re267} shows the time evolution of the kinetic
energy of the perturbation to the Poincar\'e solution, $\delta E_K
=\frac12 \|\bu - \bu_P\|_{\bL^2(\Omega)}^2$, from $t=0$ to $t=1100$,
see curve labeled ``0 perturb''. The energy grows exponentially
initially, then saturates around an oscillatory state. These results
are similar to those shown in figure 1 of~\cite{WR09}.

In order to evaluate the influence of solid rotations, we restart the
computation at $t=1100$ by adding the perturbation $\pm 0.025(-y \be_x
+ x \be_y)$ to the solution. Since the maximum norm of the Poincar\'e
solution is 1.25, the added perturbations are only 2\% of the maximum
velocity.  Time integration is performed in each case until
convergence to an oscillating state is obtained.  The
curves corresponding to the time evolution of the kinetic energy of
the solutions thus obtained are labeled ``0.025 perturb'' and
``$-0.025$ perturb'' in
Figures~\ref{fig:tot_Re267}-\ref{fig:tot_Re267:b}.  We observe in
Figure~\ref{fig:tot_Re267:b} that these perturbations have strong
impacts on the asymptotic solutions.
\begin{figure}[h]
\centerline{
\subfigure[$\delta E_K$ vs. time]{
\includegraphics[width=0.48\textwidth]{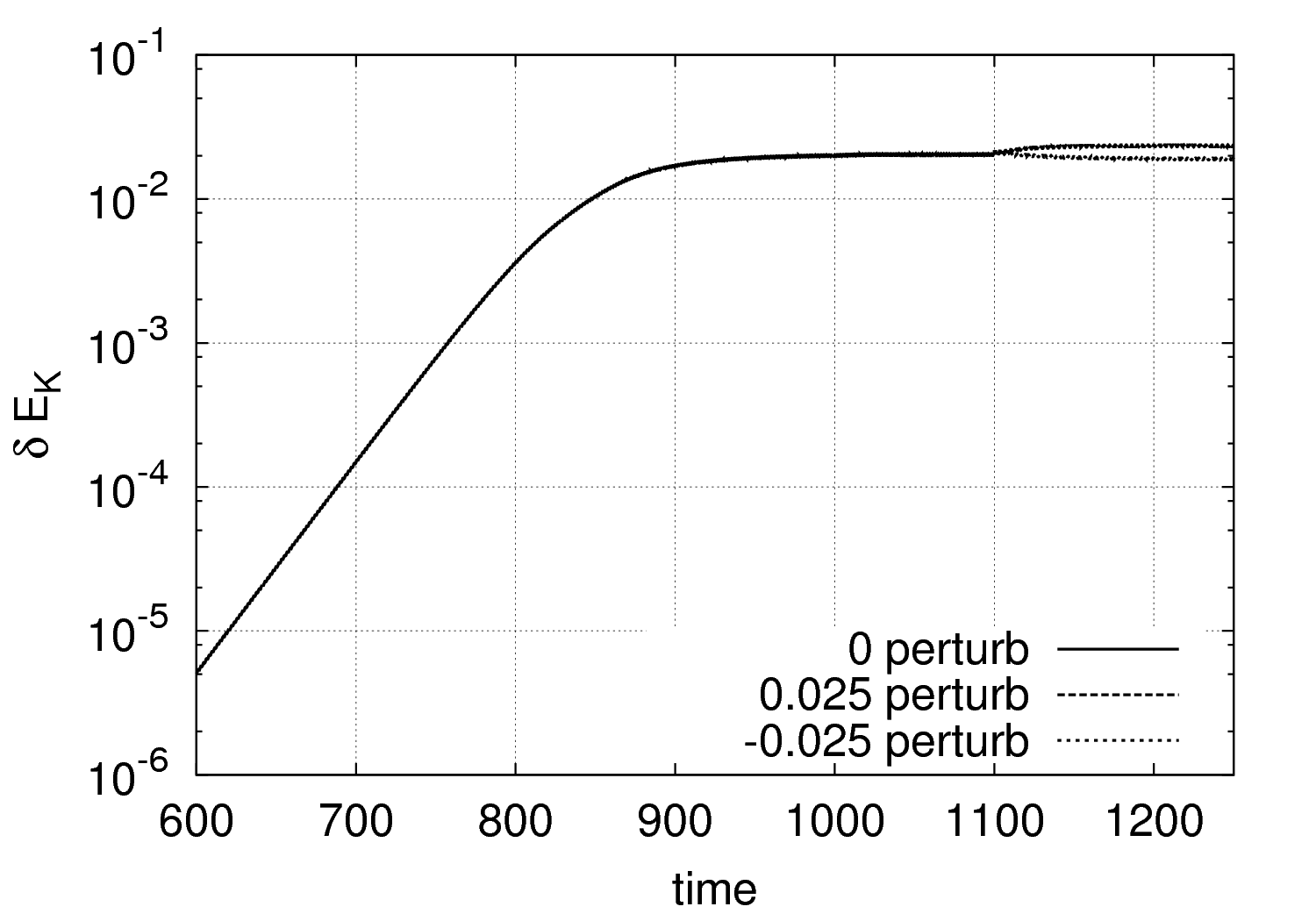}\label{fig:tot_Re267:a}}
\subfigure[zoom]{
\includegraphics[width=0.48\textwidth]{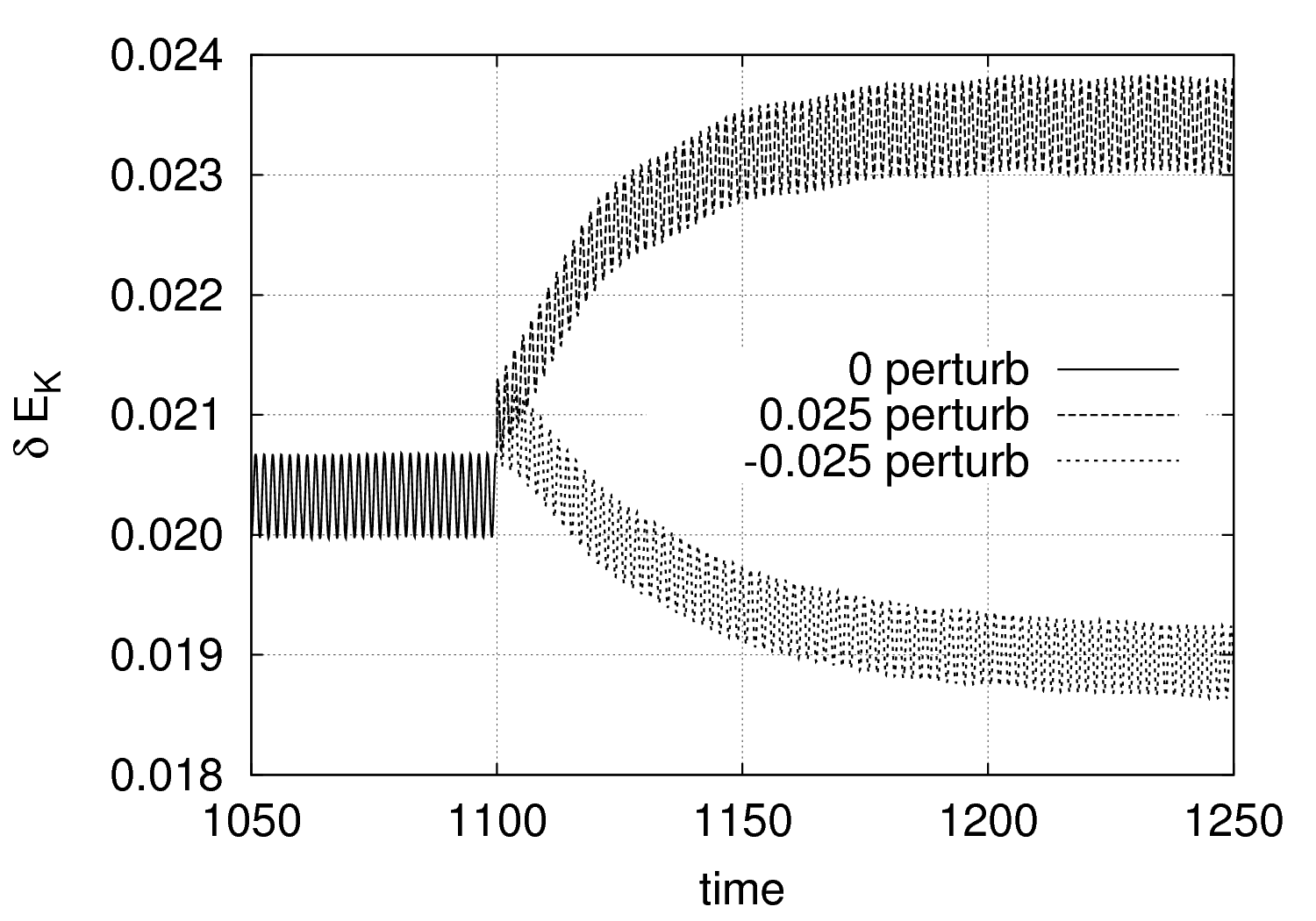}\label{fig:tot_Re267:b}}
}
\caption{(Color online) Time evolution of the kinetic energy, $\delta
  E_K$, of the perturbation of the solution of
  \eqref{eq:nsp}-\eqref{eq:divp}-\eqref{eq:slip}-\eqref{eq:Poincare_stress}
  with $\beta=0.5625$, $\epsilon=0.25$ and $1/\nu = 0.00375$ (a), and
  zoom (b).}
\label{fig:tot_Re267}
\end{figure}
\begin{figure}[h]
\centerline{
\subfigure[$\delta E_K,\, \delta E_{Kn}, \, \delta E_{Ks}$ with {\it ad hoc} initial condition]{
\includegraphics[width=0.32\textwidth]{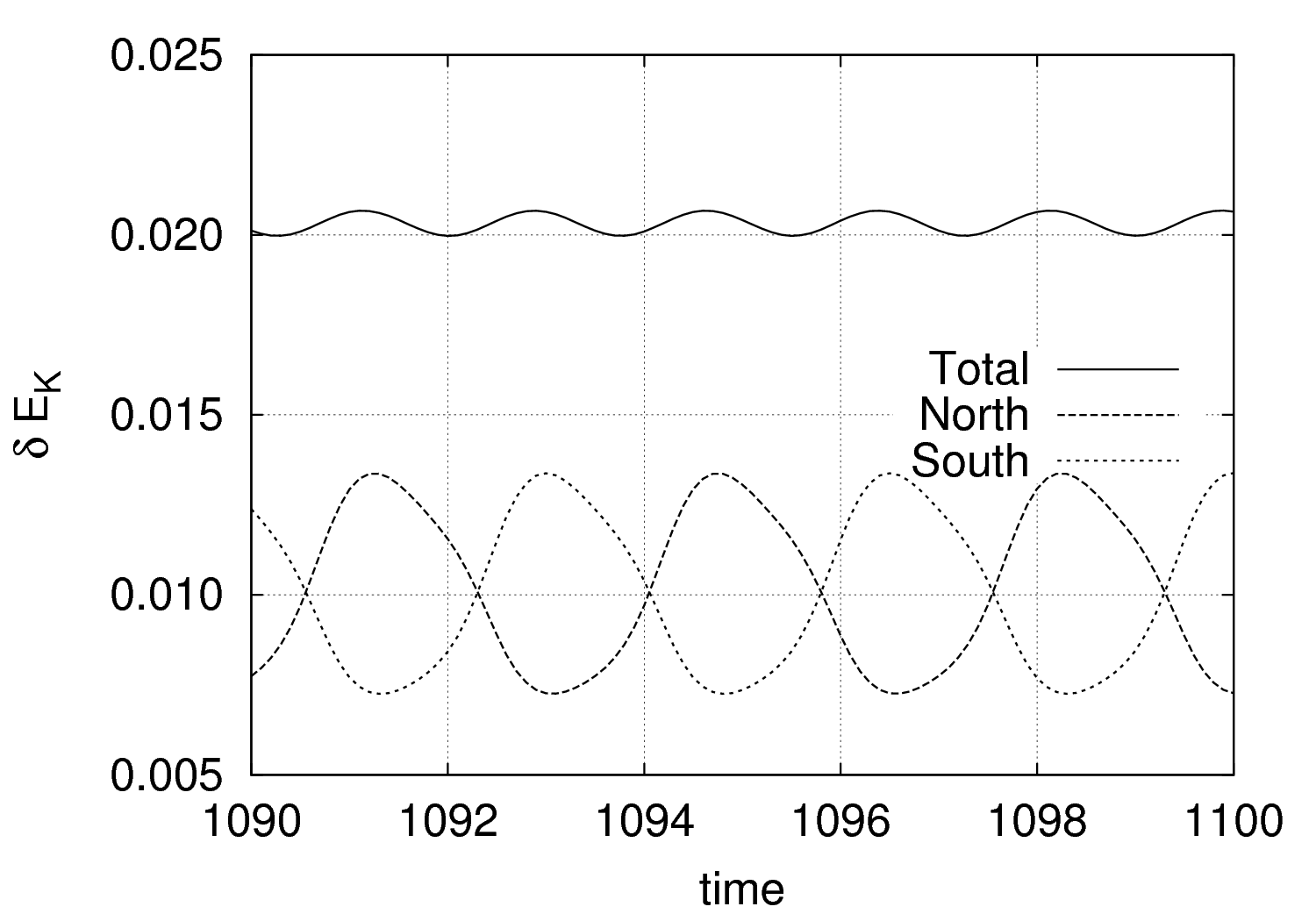}\label{fig:north_south_Re267:a}}
\subfigure[0.025 perturb]{
\includegraphics[width=0.32\textwidth]{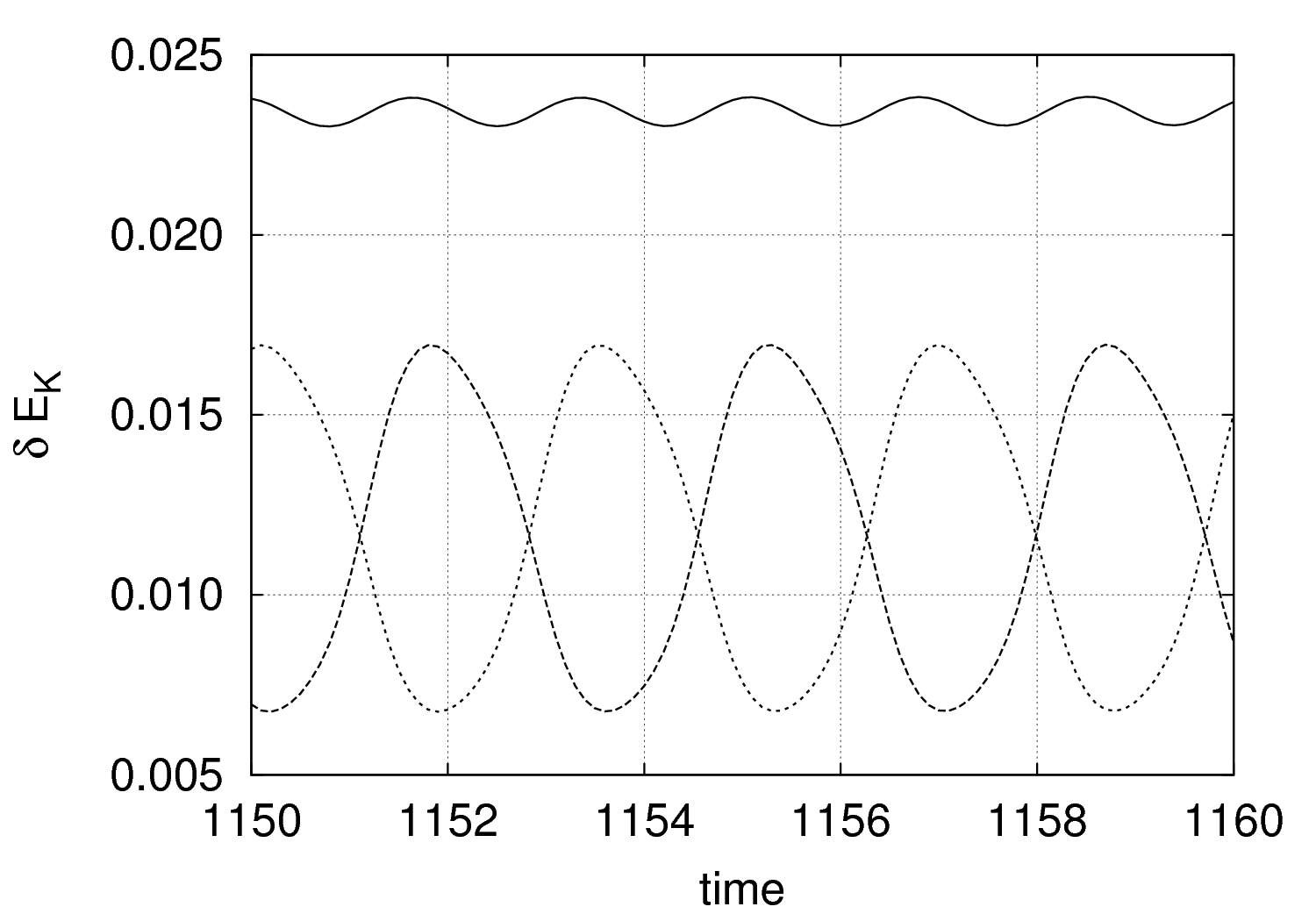}\label{fig:north_south_Re267:b}}
\subfigure[$-0.025$ perturb]{
\includegraphics[width=0.32\textwidth]{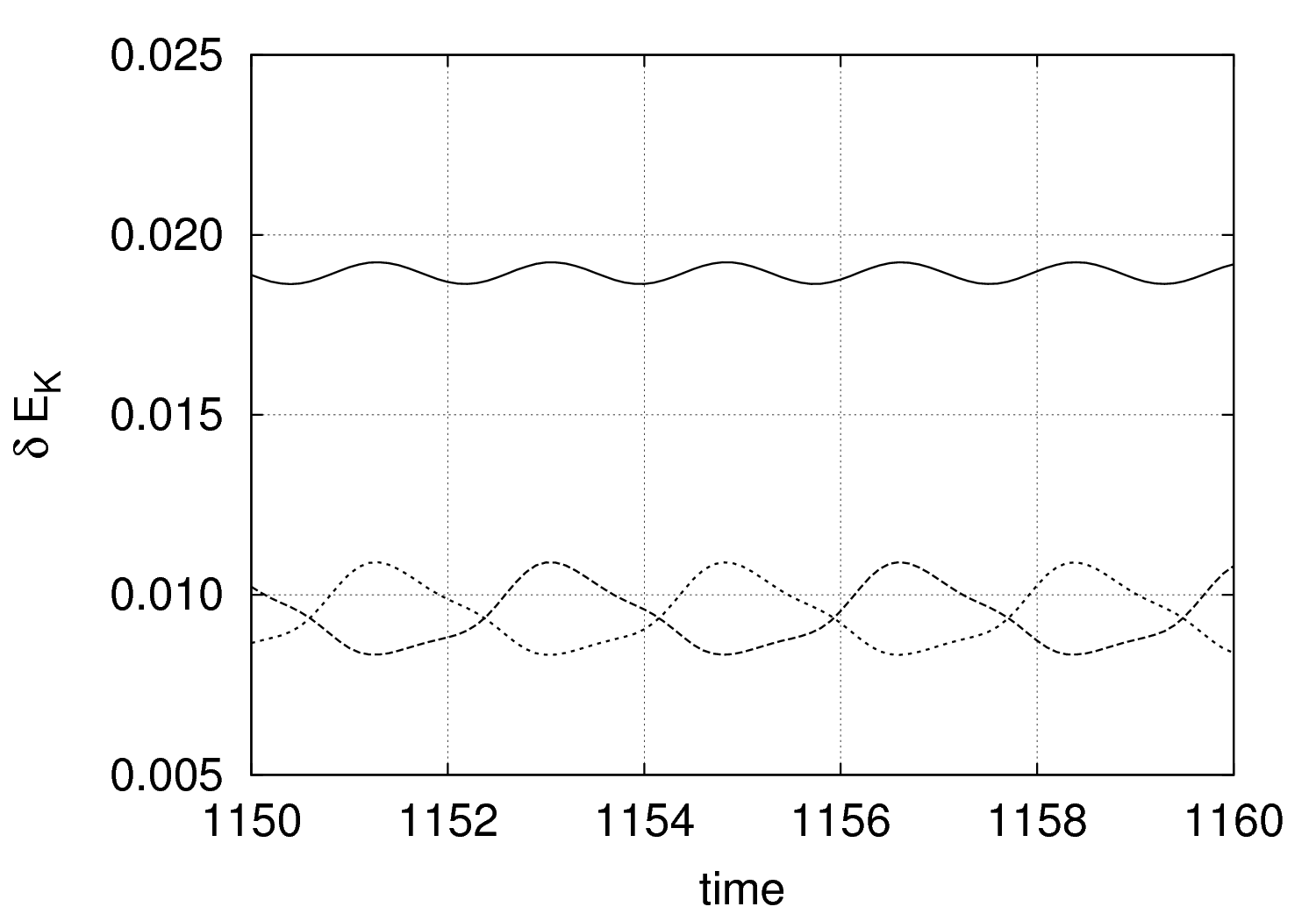}\label{fig:north_south_Re267:c}}
}
\caption{(Color online) Kinetic energy, $\delta E_K=\frac12 \|\bu -
  \bu_P\|_{\bL^2(\Omega)}^2$, where $\bu$ solves
  \eqref{eq:nsp}-\eqref{eq:divp}-\eqref{eq:slip}-\eqref{eq:Poincare_stress}
  with $\beta=0.5625,\, \epsilon=0.25$ and $1/\nu = 0.00375$: on each
  graph, top curve is $\delta E_K$, bottom curves are the energies
  $\delta E_{Kn}$ (dashed line) and $\delta E_{Ks}$ (dotted line) in
  the northern and southern hemispheres.}
\label{fig:north_south_Re267}
\end{figure}
To better compare our results with those from~\cite{WR09}, we show in
figure~\ref{fig:north_south_Re267} the energies in the northern
$\delta E_{Kn}$ and southern $\delta E_{Ks}$ hemispheres of the
spheroid.  The perturbation with the positive sign increases both the
total kinetic energy and the amplitude of the northern and southern
energies, whereas the perturbation with the negative sign decreases
both the total kinetic energy and the amplitude of the oscillations of
the northern and southern energies. The oscillations of the northern
and southern energies obtained with the positive perturbation are more
sinusoidal than those obtained with the negative perturbation.  The
shape of the oscillations of the northern and southern energies
obtained with the negative perturbation are similar to those in
\cite{WR09} (see the more pronounced nonlinear shape).
These simulations illustrate well that the 0-perturbation solution is
not an attractor, \ie it is not stable under perturbations. We have
verified (results not shown here) that an entire family of solutions
can be obtained from the 0.025-perturbation to the
$-0.025$-perturbation solutions by scaling the perturbation
appropriately. These tests show that using the stress-free boundary
condition to evaluate nonlinear behaviors of Navier-Stokes systems may
sometimes be dubious when the domain is axisymmetric.

\section{Discussion} \label{Sec:Conclusions} The so-called stress-free
boundary condition $(\bn\SCAL\bepsilon(\bu))\CROSS\bn_{|\front}=0$ is
often used in the geodynamo literature to avoid issues induced by
viscous layers. For example, very recently, an anelastic dynamo
benchmark~\cite{Jones2011120} was conducted in a rotating spherical
shell. The authors emphasize in their concluding section the
difficulties they encountered to compare four different codes using a
model with stress-free boundary conditions applied to the ICB and the
CMB. Since the container is a spherical shell, the balance
equation~\eqref{eq:moment} gives $\partial_t \bM=0$, and each group
had to apply some remedy in order to numerically conserve the three
components of the angular momentum.  But, more importantly, they also
had to use the same initial condition. There was no such difficulties
in the older dynamo benchmark \cite{Christensen200125} using the same
geometry because the no-slip boundary condition was prescribed at each
interface. These results illustrate again that the stress-free
boundary condition induces pathological stability behaviors when the
flow domain is axisymmetric.

We have shown in this work that stress-free boundary condition leads
to spurious behaviors when the fluid domain is axisymmetric. We hope
that the present work will help draw the attention of the geodynamo
community on this problem.  The above pathological stability behaviors
can be avoided by enforcing one additional condition.  For instance,
for problem~\eqref{eq:nsp}--\eqref{eq:slip} and
\eqref{eq:Poincare_stress}, one could think of enforcing the vertical
component of the angular momentum of the perturbation to the
Poincar\'e flow, say
\begin{equation} 
\int_{\front} (\bu - \bu_P)\SCAL(\be_z\CROSS\bx) \diff\bs=0, \label{u_minus_up_cross_u_P}
\end{equation}
or enforcing the perturbation and
the Poincar\'e flow to be orthogonal in average over the boundary, say
\begin{equation}
\int_{\front} (\bu - \bu_P)\SCAL \bu_P \diff\bs=0.\label{u_minus_up_dot_u_P}
\end{equation}
For problem~\eqref{eq:nssym}--\eqref{eq:nlbc}, one could think
of enforcing the vertical component of the total angular momentum
\begin{equation} 
\int_{\front} \bu \SCAL(\be_z\CROSS\bx) \diff\bs=0, \label{u_dot_ez_cross_x}
\end{equation}
as was done for the three components in the anelastic dynamo
benchmark~\cite{Jones2011120}.

We have suggested in \S\ref{Sec:admissible} to use a boundary condition 
that does not have the stability problems mentioned above. 
For the problem~\eqref{eq:nsp}--\eqref{eq:slip} this condition is
\begin{equation}
(\bn\SCAL\GRAD\bu)\CROSS\bn_{|\front} =
  (\bn\SCAL\GRAD\bu_P)\CROSS\bn_{|\front},
\end{equation}
and for the problem \eqref{eq:nssym}--\eqref{eq:nlbc} this condition is 
\begin{equation}
(\bn\SCAL\GRAD\bu)\CROSS\bn_{|\front} = 0.
\end{equation}

Let us finally emphasize that it is false to consider that the
momentum balance equation~\eqref{eq:moment} is an additional equation
that makes
\eqref{eq:nsp}-\eqref{eq:divp}-\eqref{eq:slip}-\eqref{eq:Poincare_stress}
a well-behaved dynamical system. The equation~\eqref{eq:moment} is a
redundant consequence of
\eqref{eq:nsp}-\eqref{eq:divp}-\eqref{eq:slip}-\eqref{eq:Poincare_stress}.
For instance, \eqref{u_minus_up_cross_u_P} (or
\eqref{u_minus_up_dot_u_P} or \eqref{u_dot_ez_cross_x}) is an
additional equation whereas \eqref{eq:moment} is not.

\section*{Acknowledgments} The authors are happy to acknowledge 
helpful email discussions with P.H. Roberts and P. Boronski. They also
want to thank Wietze Herreman for stimulating discussions.

%
%



%
%
%

\bibliographystyle{abbrv}   
\bibliography{bibliofl}

\begin{thebibliography}{10}

\bibitem{Bullard49}
E.~Bullard.
\newblock The magnetic field within the {E}arth.
\newblock {\em Proc. Roy. Soc. Lond. A}, 197(1051):433--453, 1949.

\bibitem{Christensen200125}
U.~Christensen, J.~Aubert, P.~Cardin, E.~Dormy, S.~Gibbons, G.~Glatzmaier,
  E.~Grote, Y.~Honkura, C.~Jones, M.~Kono, M.~Matsushima, A.~Sakuraba,
  F.~Takahashi, A.~Tilgner, J.~Wicht, and K.~Zhang.
\newblock A numerical dynamo benchmark.
\newblock {\em Physics of the Earth and Planetary Interiors}, 128(1-4):25 --
  34, 2001.
\newblock Dynamics and Magnetic Fields of the Earth's and Planetary Interiors.

\bibitem{MR1932965}
L.~Desvillettes and C.~Villani.
\newblock On a variant of {K}orn's inequality arising in statistical mechanics.
\newblock {\em ESAIM Control Optim. Calc. Var.}, 8:603--619 (electronic), 2002.
\newblock A tribute to J. L. Lions.

\bibitem{bkDuvaut}
G.~Duvaut and J.-L. Lions.
\newblock {\em {Les in\'equations en m\'ecanique et en physique}}.
\newblock Dunod, 1972.

\bibitem{GR95}
G.~A. Glatzmaier and P.~H. Roberts.
\newblock A three-dimensional self-consistent computer simulation of a
  geomagnetic field reversal.
\newblock {\em Nature}, 377:203–--209, 1995.

\bibitem{GLLNR11}
J.-L. Guermond, J.~L{\'e}orat, F.~Luddens, C.~Nore, and A.~Ribeiro.
\newblock Effects of discontinuous magnetic permeability on magnetodynamic
  problems.
\newblock {\em J. Comput. Phys.}, 230:6299--6319, 2011.

\bibitem{Jones2011120}
C.~Jones, P.~Boronski, A.~Brun, G.~Glatzmaier, T.~Gastine, M.~Miesch, and
  J.~Wicht.
\newblock Anelastic convection-driven dynamo benchmarks.
\newblock {\em Icarus}, 216(1):120 -- 135, 2011.

\bibitem{KB97}
W.~Kuang and J.~Bloxham.
\newblock An {E}arth-like numerical dynamo model.
\newblock {\em Nature}, 389:371–--374, 1997.

\bibitem{KB99}
W.~Kuang and J.~Bloxham.
\newblock Numerical modeling of magnetohydrodynamic convection in a rapidly
  rotating spherical shell: Weak and strong field dynamo action.
\newblock {\em Journal of Computational Physics}, 153(1):51 -- 81, 1999.

\bibitem{Lions69}
J.-L. Lions.
\newblock {\em Quelques m\'ethodes de r\'esolution des probl\`emes aux limites
  non lin\'eaires}, volume~1.
\newblock Dunod, Paris, France, 1969.

\bibitem{Malkus68}
W.~V.~R. Malkus.
\newblock Precession of the earth as the cause of geomagnetism.
\newblock {\em Science}, 160(3825):259--264, 1968.

\bibitem{MK02}
R.~Mason and R.~Kerswell.
\newblock Chaotic dynamics in a strained rotating flow: a precessing plane
  fluid layer.
\newblock {\em J. Fluid Mech.}, 471:71--106, 2002.

\bibitem{NLGL11}
C.~Nore, J.~L\'eorat, J.-L. Guermond, and F.~Luddens.
\newblock Nonlinear dynamo action in a precessing cylindrical container.
\newblock {\em Phys. Rev. E}, 84:016317, Jul 2011.

\bibitem{Olson97}
P.~Olson.
\newblock Probing {E}arth's dynamo.
\newblock {\em Nature}, 389:337--338, 1997.

\bibitem{tilgner_precession_2005}
A.~Tilgner.
\newblock Precession driven dynamos.
\newblock {\em Physics of Fluids}, 17(3):034104, 2005.

\bibitem{tilgner_kinematic_2007}
A.~Tilgner.
\newblock Kinematic dynamos with precession driven flow in a sphere.
\newblock {\em Geophysical \& Astrophysical Fluid Dynamics}, 101(1):1, 2007.

\bibitem{WR09}
C.-C. Wu and P.~Roberts.
\newblock On a dynamo driven by topographic precession.
\newblock {\em Geophysical \& Astrophysical Fluid Dynamics}, 103(6):467--501,
  2009.

\end{thebibliography}

\end{document}